\documentclass[11pt,a4paper]{amsart}
\usepackage{amsfonts}
\usepackage{amssymb}
\usepackage{amsthm}
\usepackage[utf8]{inputenc} 
\usepackage{graphicx}

\theoremstyle{plain}
\newtheorem{thm}{Theorem}[section]

\newtheorem{cor}[thm]{Corollary}
\newtheorem{prop}[thm]{Proposition}
\theoremstyle{definition}

\theoremstyle{remark}

\begin{document}

\title[Comparison of Stopping Times by Nested CMC]{Faster Comparison of Stopping Times by Nested Conditional Monte Carlo}

\author{Fabian Dickmann}
\address{Faculty of Mathematics\\ Duisburg-Essen University\\
Thea-Leymann-Str. 9\\ 45127 Essen\\ Germany}
\email{fabian.dickmann@uni-due.de} 
\thanks{Both authors gratefully acknowledge financial support by the Deutsche Forschungsgemeinschaft through SPP 1324.}

\author{Nikolaus Schweizer}
\address{Department of Mathematics\\ Saarland University\\
Postfach 151150\\ 66041 Saarbr\"ucken\\ Germany}
\email{schweizer@math.uni-sb.de }
\date{January 2014}





\begin{abstract}
We show that deliberately introducing a nested simulation stage can lead to significant variance reductions when comparing 
two stopping times by Monte Carlo. We derive the optimal number of nested simulations and prove that the algorithm is remarkably
robust to misspecifications of this number. The method is applied to several problems related to Bermudan/American options. 
In these applications, our method allows to substantially increase the efficiency of other variance reduction techniques, namely, Quasi-Control Variates and 
Multilevel Monte Carlo.
\end{abstract}
\keywords{American Options, Bermudan Options, Branching, Importance Sampling, Multilevel Monte Carlo, Nested Simulation, Optimal Stopping, Splitting, Variance Reduction}

 \maketitle

 \section{Introduction}
 
In this paper, we propose a novel method for efficiently comparing the performance of different stopping times, i.e. we are interested in computing $\Delta=E[X_{\tau^A}-X_{\tau^B}]$ where
 $X$ is a stochastic process in discrete time and $ \tau^A$ and $\tau^B$ are two stopping times. A simple Monte Carlo algorithm for this problem consists of simulating $N$ trajectories 
 of $X$ until both stopping times have occurred and then to average over the resulting $N$ realizations of $X_{\tau^A}-X_{\tau^B}$. If $\tau^A$ and $\tau^B$ are similar, e.g., because they are
 two approximations of the same intractable stopping time, or solutions to two similar problems, this method tends to be inefficient, because sizeable contributions to $\Delta$ come only from
 the few regions in state space where $\tau^A$ and $\tau^B$ disagree.
 
 Instead, we write 
 \begin{equation}\label{twostage}
 \Delta=E[ E[ X_{\tau^A}-X_{\tau^B} |\mathcal{F}_{\tau^\wedge} ]],\;\;\;  \text{ where }\;\;\; \tau^\wedge = \min(\tau^A,\tau^B),
 \end{equation}
where $\mathcal{F}_{ \tau^\wedge}$ denotes the information generated until the first of the stopping times occurs and propose the following two-stage simulation procedure: Simulate $N$
 trajectories of $X$ until the first stopping time occurs, i.e., until $\tau^\wedge$. Then simulate $R$ conditionally independent copies of each of the $N$ trajectories until 
 the second stopping time $\tau^\vee =\max(\tau^A,\tau^B)$ and estimate $\Delta$ by the mean of the $R\cdot N$ realizations of $X_{\tau^A}-X_{\tau^B}$. The resulting estimator can be interpreted as
  estimating first for each of the $N$ initial trajectories the inner conditional expectation in \eqref{twostage} by the mean over the $R$ replications of that trajectory, and then averaging 
 over the $N$ initial trajectories to estimate the outer expectation.
 
 The idea of using subsimulations to estimate an inner conditional expectation relates our approach to the literature on Nested Simulation \cite{gordy2010nested, broadie2011efficient}. 
 In the applications considered there, there is a non-linear dependence on the inner conditional expectation so that the inner simulations are indispensable
 -- and are generally regarded as an unavoidable burden. From this perspective, it is interesting to see that in our numerical examples, where inner simulations are introduced 
 deliberately as a variance reduction technique, the estimated optimal numbers of inner paths, do not differ much from what is typically used in these applications,  e.g. $R=100$.
Variance reduction by deliberately inserting a conditional expectation is a common technique if these conditional expectations are available in closed form. This classical method is known
as Rao-Blackwellization or Conditional Monte Carlo \cite{boyle1997monte, asmussen1997simulation}. It is usually not applicable in our setting since closed-form expressions for expectations of stopped processes are rare.
Since our method mimics Conditional Monte Carlo by Nested Simulation, we refer to it as \textit{Nested Conditional Monte Carlo}.

We apply our method to three problems related to Bermudan option pricing.\footnote{ We refer to the options as Bermudan, since the sets of exercise times are finite. We could also interpret them as a time discretizations of American options.}
In the first application, there is a genuine interest in comparing different stopping times: We consider the impact of parameter uncertainty on the performance of estimated optimal exercise strategies.
In the two other applications, the true quantity of interest is  $E[X_{\tau^A}]$ while $X_{\tau^B}$ serves as a control variate. Our method is used
to enhance the efficiency of this control variate. For two variance reduction methods of this type, a Quasi-Control Variate as introduced in \cite{ES}
and the Multilevel algorithm of \cite{belomestny2013pricing}, we demonstrate that Nested Conditional Monte Carlo can lead to sizeable efficiency gains. In fact, for both algorithms
the additional variance reduction through incorporating our method is at least as large as the variance reduction of the original algorithm.


  In a sense, our algorithm closely resembles the splitting algorithms for rare event simulation studied, e.g., in \cite{villen2002analysis, glasserman1999multilevel}.
 In the applications considered in this literature (e.g. barrier option pricing or estimating the probability of large losses),
 the rare event typically consists of $X$
 taking exceptionally large  or small values. Thus, the trigger events for replicating a trajectory are chosen as the hitting times of some threshold value of $X$. In this way, computational effort can 
be allocated efficiently to the regions  where it is needed the most.

 Our rare event, a large discrepancy between $X_{\tau^A}$ and $X_{\tau^B}$, does not have such a nice structure, i.e., it is not easy to connect it a priori to particular values of $X$ which
 might serve as a trigger for replications.
 Our key observation is that this type of knowledge is not necessary here: We simply start replications at an event, the first stopping time, which occurs once on \textit{every} trajectory. 
 An efficient allocation of computational effort to critical trajectories occurs endogenously by the following reasoning: $X_{\tau^A}-X_{\tau^B}$ is expensive to simulate
 if $\tau^A$ and $\tau^B$ lie far apart. Yet those cases where $\tau^A$ and $\tau^B$ lie far apart also carry a substantial probability for large values of $X_{\tau^A}-X_{\tau^B}$.
 In addition, in many cases the difference $X_{\tau^A}-X_{\tau^B}$ will not depend strongly on what happens before $\tau^\wedge$. Thus, we gain efficiency by shifting computational effort 
 from the time interval $[0, \tau^\wedge]$ to the interval $[\tau^{\wedge},\tau^\vee]$ (which is typically much shorter). 
 To sum up, we use splitting with the objective of identifying important regions in time rather than in space -- although, of course,  the two cannot fully be disentangled.
 On a more abstract level, our results demonstrate that ideas from rare event simulation  can be used to boost the efficiency of control variates even when dealing with
 ``ordinary'' events.
 As we will see in the numerical examples, this gives a highly efficient Monte Carlo algorithm which has only one free parameter, the number of replications $R$.
 
 A particular advantage of splitting methods such as ours  is that they result in unbiased estimators. This sets them apart from related algorithms such as Importance Sampling
 and particle methods (see the survey \cite{carmonaintroduction})  which also aim at an endogenous and efficient distribution of simulation costs.
 Unbiasedness is of particular importance in option pricing applications, where one wishes to calculate estimators which are known to have a positive or negative 
 bias in order to construct confidence intervals, see e.g. \cite{andersen2004primal}.

 The paper is organized as follows: Section \ref{secAlgo} introduces the setting and the algorithm and derives a formula for the variance of the estimator.
 Section \ref{secCali} characterizes situations where our algorithm leads to significant improvements over simple Monte Carlo and derives a formula for the optimal 
 number of replications $R$. Moreover, we prove that the performance of the estimator is rather insensitive to moderate misspecifications of $R$: As long as $R$ is, say, within 
 20\% of the optimal value we achieve more than 99\% of the optimal variance reduction. Section \ref{secApp} demonstrates the algorithm's efficiency in the three option pricing applications
 sketched above. Section \ref{secCon} concludes. All proofs are in the appendix.

 \section{The Algorithm}\label{secAlgo}
 Consider a square-integrable, adapted, real-valued stochastic process $(X_j)_j$ on a complete filtered probability space $(\Omega, \mathcal{F}, (\mathcal{F}_j)_{j=0}^J,P )$ over the discrete time
 horizon $\{0,1,\ldots,J\}$. There are two stopping times $\tau^A$ and $\tau^B$ and we are interested in computing
 \[
\Delta = E[X_{\tau^A}-X_{\tau^B}]
 \]
 by a Monte Carlo approach. We define the stopping times $\tau^\wedge= \min(\tau^A, \tau^B)$ and $\tau^\vee= \max(\tau^A, \tau^B)$ and the random variable $S=\textnormal{sign}(\tau^A-\tau^B)$ and note that 
 \[
 \tau^A-\tau^B = S(\tau^\vee-\tau^\wedge), \;\;\;\textnormal{that }  \;\;\;   X_{\tau^A}-X_{\tau^B} = S(X_{\tau^\vee}-X_{\tau^\wedge})
 \]
and that $S$ is observable at time $\tau^\wedge$. For our Monte Carlo approach, we assume that (conditionally) independent copies of random variables are available as needed on our probability space
and propose the following two-stage simulation algorithm which is determined by two integer-valued, positive parameters $N$ and $R$:
\begin{itemize}
  \item[A1.] Simulate independent copies $X_0^{(i)},\ldots, X_{\tau^{\wedge, (i)}}^{(i)}$  of $X_0,\ldots, X_{\tau^{\wedge}}$ for $i=1,\ldots N$.  
  Denote by $\mathcal{F}^{\tau^\wedge,(i)}$ the information generated along the $i^{th}$
  trajectory and by $S^{(i)}$ the associated copy of $S$.
  \item[A2.] Conditionally on $\mathcal{F}^{\tau^\wedge,(i)}$ simulate for each $i$ with $S^{(i)}\neq 0$ and for $r=1,\ldots, R$ copies 
  $X^{(i,r)}_{\tau^{\wedge, (i)}+1},\ldots, X^{(i,r)}_{\tau^{\vee, (i,r)}}$ of $X_{\tau^{\wedge}+1},\ldots, X_{\tau^{\vee}}$
  which are independent across  the $i$ and conditionally independent across the $r$. If $S^{(i)}=0$ and thus $\tau^{\vee,(i,r)}=\tau^{\wedge,(i)}$ set $X^{(i,r)}_{\tau^{\vee, (i,r)}}= X_{\tau^{\wedge, (i)}}^{(i)}$.
  \item[A3.] Estimate $\Delta$ by
  \begin{equation}
  \Delta^{(N,R)}=\frac{1}{N} \sum_{i=1}^N \frac{1}{R} \sum_{r=1}^R S^{(i)} (X^{(i,r)}_{\tau^{\vee,(i,r)}} - X^{(i)}_{\tau^{\wedge,(i)}})
    \end{equation}
\end{itemize}

In short, we simulate $N$ independent trajectories of $X$ until the first stopping time $\tau^\wedge$. From $\tau^\wedge$ on, we simulate $R$ copies of each trajectory until the second 
stopping time $\tau^\vee$. $\Delta$ is estimated by the mean $\Delta^{(N,R)}$ of the $R \cdot N$ (dependent) samples of $X_{\tau^A}-X_{\tau^B}$ obtained in this way.

Figure 1 illustrates the simulation procedure for an example where the two stopping times are given in the form of exercise boundaries, the blue and red curves in the figure.
Whenever the process $X$ crosses either boundary, the stopping time occurs. From the small amounts of red and blue in the picture, we observe that the subsampling in Step A2 only has to be carried out rarely -- when the 
two stopping times differ, implying that the additional 
simulation costs from subsampling tends to be small. 

As will become clear below, $\Delta^{(N,R)}$ can be interpreted as an approximation of a so-called Conditional Monte Carlo estimator where conditional expectations are replaced
by nested simulations (subsamples). We thus refer to the algorithm as the \textit{Nested Conditional Monte Carlo} algorithm.
The following proposition shows that $\Delta^{(N,R)}$ is unbiased and gives an expression for its variance.
\begin{prop}\label{propUB}
 We have $E[\Delta^{(N,R)}]=\Delta$ and
 \[
 \textnormal{Var}(\Delta^{(N,R)})=\frac{v_1}{N} + \frac{v_2}{R N}
 \]
 where 
 \[
 v_1 = \textnormal{Var}(E[X_{\tau^A}-X_{\tau^B}| \mathcal{F}_{\tau^{\wedge}}]) \;\; \textnormal{ and }\;\; v_2 = E[\textnormal{Var}(X_{\tau^A}-X_{\tau^B}| \mathcal{F}_{\tau^{\wedge}})].
 \]
\end{prop}

The basic motivation for the algorithm is as follows: If $\tau^A$ and $\tau^B$ are not far apart, Step A2 of the algorithm is much cheaper computationally than Step A1. If they happen to coincide, Step A2 is for free. Therefore, large values of $R$ are comparatively cheap.
Moreover, if  circumstances are favorable, namely, if the bulk of the variance in $X_{\tau^A}-X_{\tau^B}$ actually comes from what happens between $\tau^A$ and $\tau^B$, i.e., if $v_1 \ll v_2$, the estimator will behave like an estimator with $R \cdot N$ rather than $N$
samples.

\begin{figure}
\centerline{\rotatebox[origin=c]{0}{
\includegraphics[scale=0.45]{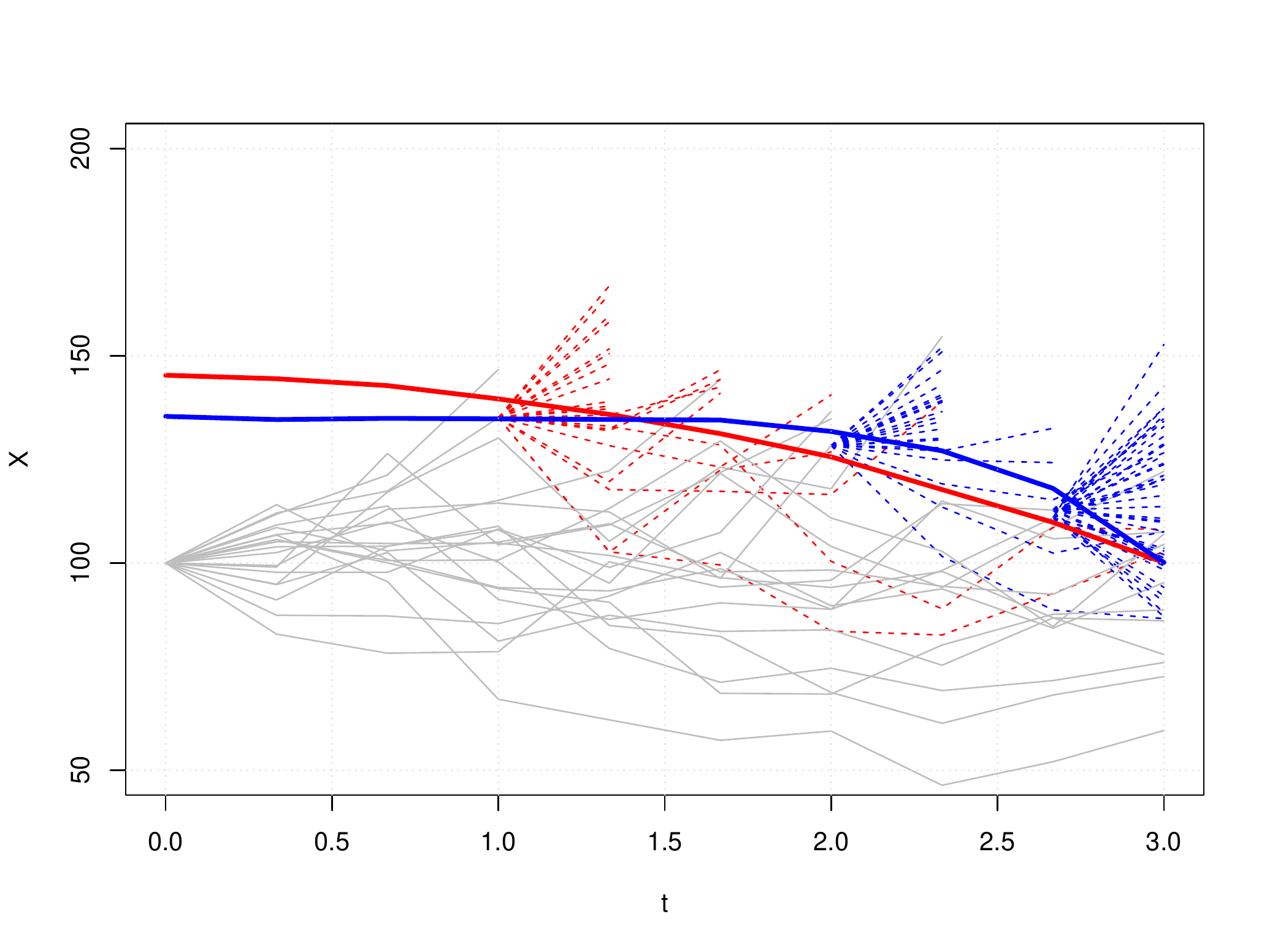}}}
\caption{Simulated trajectories stopped and replicated at two exercise boundaries.}
\end{figure}

We close this section by pointing out that $\Delta^{(N,R)}$ can be understood as an interpolation between two well-known Monte Carlo algorithms: $R=1$ corresponds to a simple Monte Carlo estimator
and $R=\infty$ corresponds to Conditional Monte Carlo.
For $R=1$, the algorithm collapses to a simple Monte Carlo estimator $\Delta^{MC}=\Delta^{(N,1)}$ of $X_{\tau^A}-X_{\tau^B}$ along $N$ sample trajectories of $X_0,\ldots, X_{\tau^\vee}$. Moreover, we have
\[
 \textnormal{Var}(\Delta^{ MC})=\frac{\textnormal{Var}(X_{\tau^A}-X_{\tau^B})}{N} = \frac{v_1+v_2}{N}.
\]
 The final equality is the well-known conditional variance decomposition formula. The inner sum in the estimator $\Delta^{(N,R)}$,
\[
\frac{1}{R} \sum_{r=1}^R S^{(i)} (X^{(i,r)}_{\tau^{\vee,(i,r)}} - X^{(i)}_{\tau^{\wedge,(i)}}),
\]
can be interpreted as a Monte Carlo estimator for
\[
D^{(i)} = E[S^{(i)}\,(X^{(i,1)}_{\tau^{\vee,(i,1)}} - X^{(i)}_{\tau^{\wedge,(i)}}) | \mathcal{F}_{{\tau^{\wedge,(i)}}} ].
\]
which is exact in the limit $R \rightarrow  \infty$. The limiting Monte Carlo estimator
\[
\Delta^{ CMC}=\frac{1}{N} \sum_{i=1}^N D^{(i)}
 \]
is the so-called Conditional Monte Carlo (CMC) estimator  for $\Delta$. 
In the applications we consider, the conditional expectation in $D^{(i)}$ typically cannot be computed explicitly and thus the estimator $\Delta^{ CMC}=\Delta^{(N,\infty)}$
is purely of theoretical interest.
The variance of $\Delta^{CMC}$ is given by
\[
 \textnormal{Var}(\Delta^{CMC})=\frac{v_1}{N},
\]
thus reducing the variance by a factor $v_1/(v_1+v_2)$ compared to $\Delta^{MC}$. 
By employing nested simulation ($R>1$)  to approximate the conditional expectations, we construct implementable estimators which achieve at least part of this variance reduction.

 \section{Calibrating the Algorithm}\label{secCali}
 In the previous section, we saw that $\Delta^{(N,R)}$ achieves a guaranteed variance reduction compared to the simple Monte Carlo estimator $\Delta^{(N,1)}$. 
 This comparison is however unfair since increasing the number of replications $R$ leads to higher computational costs.  
In this section, we thus address the following set of closely related questions: Given a fixed computational budget, when is it favorable to implement the estimator $\Delta^{(N,R)}$ 
 with $R>1$? How can we determine the optimal value of $R$? How much do we gain compared to the simple Monte Carlo method with $R=1$? How sensitive is the algorithm's performance to
 miscalibration?
 
 It turns out that the answers to all these questions crucially
 depend on two natural conditions: Implementing $\Delta^{(N,R)}$ with $R>1$ can lead to drastic computational efficiency gains 1) if the theoretical Conditional Monte Carlo estimator $\Delta^{ CMC}$ 
would lead to a substantial variance reduction, and 2) if generating samples of $X_{\tau^{\vee}} - X_{\tau^{\wedge}}$ conditionally on $X_{\tau^{\wedge}}$ is cheaper than sampling copies of $X_{\tau^{\wedge}}$,
i.e., if the cost of a single sample is smaller in Step A2 of the algorithm than in Step A1. If $\tau^A$ and $\tau^B$ are sufficiently similar, 
the second condition is likely to hold, i.e., the difference between the stopping times will be much smaller in expectation than their minimum. 
 
 Observe next that the computational costs of implementing the estimator for fixed $N$ and $R$ is itself random since the number of time steps that have to be simulated depends
 on the realizations of the stopping times. This is the case even for the simple Monte Carlo estimator with $R=1$. The following analysis is thus based on expected computational costs.
 
 Denote by $\rho_1$ the expected computational cost of simulating a realization of $X_{\tau^{\wedge}}$ in Step A1 of the algorithm. $\rho_1$ takes into account the expected length   
 $E[\tau^{\wedge}]$ of a path and the costs of evaluating both stopping times along that path. 
 Denote by $\rho_2$ the expected computational cost of simulating a realization of $X_{\tau^{\vee}}-X_{\tau^{\wedge}}$ for given $X_{\tau^{\wedge}}$ in Step A2 of the algorithm.
 $\rho_2$ takes into account the expected length   
 $E[\tau^{\vee}-\tau^{\wedge}]$ of such a path and the costs of evaluating either of the two stopping times along that path. The expected computational cost of implementing the estimator
 with parameters $N$ and $R$ is thus given by
 \[
c(N,R)= N \rho_1 + N R \rho_2.
 \]
 For simplicity, we assume that $v_1$, $v_2$, $\rho_1$ and $\rho_2$ are strictly positive.
 
The next proposition characterizes how to optimally choose $N$ and $R$ for a given computational budget $C$. 
In particular, we derive an expression for the optimizer $R^*$ which is independent of the
overall budget, showing that in relative terms the optimal allocation of computational costs between Steps A1 and A2 is independent of $C$.
For this reason, the question of calibrating the algorithm is basically reduced to finding a good choice of $R$.
For the moment, we ignore the integer-constraints on $N$ and $R$. We
do however take into account that in order to obtain an implementable algorithm we must have $R \geq 1$. By identifying the situations where $R^* >1$ we identify those cases
where our Nested Conditional Monte Carlo algorithm is more efficient than simple Monte Carlo.

\begin{prop}\label{propRstern}
For any $C>0$, the solution $(N^*, R^*)$ to 
\[
\min_{N,R} \textnormal{Var}(\Delta^{(N,R)})\;\; \;\;\textnormal{s.t.} \;\;\;\;c(N,R)\leq C,\; R\geq 1,\; N\geq 0
 \]
 is given as follows: If 
 \begin{equation}\label{cond}
  \frac{\rho_1}{\rho_2} \frac{v_2}{v_1} > 1
 \end{equation}
then 
\[
R^* =\sqrt{\frac{\rho_1}{\rho_2} \frac{v_2}{v_1}} \;\;\text{ and} \;\;N^*= \frac{C}{\rho_1+\rho_2 R^*}.
\]
 If condition \eqref{cond} is violated, the optimal choice is
\[
R^* =1 \;\;\text{ and} \;\;N^*= \frac{C}{\rho_1+\rho_2 }.
\]
\end{prop}

From condition \eqref{cond} we can characterize the cases where the algorithm with $R>1$ is preferable to simple Monte Carlo as follows: \eqref{cond} is fulfilled if the cost of a single sample
is smaller in Step A2 than in Step A1, $\rho_2 < \rho_1$, and if a perfect CMC estimator would reduce the variance by at least a factor 2, $v_1< v_2$. If either of these conditions fails,
 \eqref{cond} can only hold if the other condition is satisfied sufficiently strongly. 
 
 If $N$ and $R$ are such that the budget constraint $c(N,R)=C$ holds with equality then the variance of $\Delta^{(N,R)}$ can be written as
 \[
 \textnormal{Var}(\Delta^{(N,R)})= \frac{V(R)}{C}\;\; \text{where} \;\;V(R)= (\rho_1 +\rho_2 R )\left(v_1 +  \frac{v_2}{R}\right).
 \]
 Therefore, in order to compare the resulting variance across different values of $R$ it suffices to compare the values $V(R)$. 
 The next proposition quantifies the gain from using our algorithm with $R^*$ subsamples rather than a simple
 Monte Carlo estimator: 
 
\begin{prop}\label{propgain}
If condition \eqref{cond} holds, the relative gain (variance reduction) from using Nested Conditional Monte Carlo with $R^*$ subsamples instead of simple Monte Carlo is given by
\[
\gamma^*= \frac{V(R^*)}{V(1)}=\frac{\left(\sqrt{\frac{v_1}{v_2}}+\sqrt{\frac{\rho_2}{\rho_1}}\,\right)^2}{(1+\frac{v_1}{v_2})(1+\frac{\rho_2}{\rho_1})}.
\]
Moreover,
\[
 \max\left(\frac{\rho_2}{\rho_1+\rho_2}, \frac{v_1}{v_1+v_2} \right) \leq \gamma^* \leq 4\, \max\left(\frac{\rho_2}{\rho_1+\rho_2}, \frac{v_1}{v_1+v_2} \right).
\]
\end{prop}

The lower bound on $\gamma^*$ shows that the variance parameters $v_i$ and the cost parameters $\rho_i$ independently place a bound on the variance reduction we can hope to achieve:
We can reduce the variance at most by a factor $v_1/(v_1+v_2)$, no matter how small $\rho_2$ is compared to $\rho_1$. The intuitive reason for this is that $\Delta^{(N,R)}$
can never beat the theoretical CMC estimator $\Delta^{CMC}$. Likewise, no matter how small the CMC-variance $v_1$ is compared to $v_2$, we can never gain more than the speed-up from 
concentrating our simulations on the interval from $\tau^\wedge$ and $\tau^\vee$ instead of the whole interval from $0$ to $\tau^\vee$. This speed-up is captured by the ratio between
$\rho_2$ and $\rho_1+\rho_2$. Since our upper bound on $\gamma^*$ is four times the lower bound, we see that the lower bound is never too far off.
To sum up, we can expect drastic variance reductions if (and only if) $v_1 \ll v_2$ and $\rho_2 \ll \rho_1$.

In practical implementations, we will not be able to work with exactly $R^*$ subsamples for at least two reasons: Since we will not know the parameters $v_1$, $v_2$, $\rho_1$ and $\rho_2$, 
these have to be estimated in pilot simulations. Moreover, $R$ has to be set to an integer value. Thus, it is important to make sure that the performance of the algorithm is not too
sensitive to the choice of $R$. The next proposition shows that this is indeed the case, giving an upper bound on the loss in variance reduction if we can only guarantee that $R$ lies in an interval
around $R^*$.

\begin{prop}\label{propAlpha}
Suppose that $R^*>1$ and $\alpha^{-1} R^* \leq  R \leq \alpha R^*$ for some $\alpha>1$. Then we have the following bound on the loss in variance reduction:
\[
\frac{V(R)}{V(R^*)} \leq \frac{1}{2}+\frac{\alpha+\alpha^{-1}}{4}, \; \; \;\text{and thus}  \; \;\; \frac{V(R)}{V(1)} \leq \left(\frac{1}{2}+\frac{\alpha+\alpha^{-1}}{4}\right) \, \gamma^*.
\]
 \end{prop}
 
This bound is fairly tight for realistic values of $\alpha$. For $\alpha=1.2$, implying that $R$ is misspecified by about 20\%, we are still within 1\% of the optimal
variance reduction. For $\alpha=2$, almost 90 \% of the optimal variance reduction are achieved. We thus conclude that even a crude attempt at optimizing the number of subsamples $R$
should lead to near-optimal results. 

A key observation in the proof of Proposition \ref{propAlpha} is the identity $V(\alpha R^*)=V(\alpha^{-1} R^*)$.
The next corollary collects some of its practical implications for the choice of $R$: If $R^*$ is significantly larger than $1$, then there is a wide interval of values for $R$
which give an improvement over simple Monte Carlo: Any value of $R$ which is smaller than the square of the optimum $R^*$ is better than $R=1$. 
Moreover, given a fixed computational budget it is always better to overestimate $R^*$ by a fixed amount, than to underestimate it by the same amount. 
Finally, rounding $R^*$ to the nearest integer can never produce an algorithm which is worse than simple Monte Carlo.

\begin{cor}\label{propR} Suppose condition \eqref{cond} holds, i.e., $R^* > 1$. Then the following assertion are true:
\item[(i)]
For every $R$ with $1 < R < {R^*}^2$ we have an improvement over simple Monte Carlo, $V(R)< V(1)$.
\item[(ii)] Let $r>0$ be such that $R^*-r \geq 1$. Then $V(R^*+r)< V(R^*-r)$.
\item[(iii)] Let $R^{\#}$ be the integer nearest to $R^*$. If $R^\#>1$, then $V(R^\#)<V(1)$.
\end{cor}


 \section{Applications to Bermudan Option Pricing }\label{secApp}
In this section, we illustrate our algorithm in a number of applications related to a well-known benchmark example \cite{andersen2004primal, BroGla04} from Bermudan option pricing, the valuation of
a Bermudan max-call option in a Black-Scholes model with dividends. There are $d$ stocks with price processes $Y_t^d$ over the continuous time horizon $[0,T]$. Under the
risk-neutral pricing measure, the stocks are independent, identically distributed geometric Brownian motions with volatility $\sigma$ and drift $r-\delta$. Here, $r$ is the risk-free interest rate
and $\delta$ is the dividend yield. Assume that there is a finite, ordered set of exercise dates $t_0, t_1, \ldots, t_J$ in $[0,T]$ and write
\[
X_j = e^{-r\, t_j}\left( \max_d Y_{t_j}^d -K\right)^+ 
.
\]
Thus, $X_j$ is the discounted payoff from exercising a Bermudan max-call option with strike $K$ at time $t_j$. The fair price at time 0 of this option is given by $E[X_{\tau^*}]$ where the 
optimal stopping time $\tau^*$ solves
\[
\sup_{\tau} E[X_{\tau}].
\]
Here, the supremum runs over all stopping times with values in $\{0,\ldots,J\}$. 

This stopping problem is numerically intricate unless the dimension $d$ is small. One popular method for 
identifying confidence intervals for the fair price is the primal-dual approach put forward in \cite{andersen2004primal}. First, one calculates an approximation $\tau$ of $\tau^*$ and estimates
$E[X_\tau]$ which gives a lower bound due to the suboptimality of $\tau$. Afterwards, $\tau$ can also be used to construct high-biased estimators, relying on the dual approach of \cite{haugh2004pricing, rogers2002monte}.
We focus here on the first step of approximating $\tau$ and evaluating the low-biased estimates. 

Note that if we use Monte Carlo methods for both the construction of $\tau$
and the evaluation of $E[X_\tau]$, we need to use independent randomness in the two calculations to preserve the low-biasedness property, see the discussion in \cite{Gl}. Throughout,
we refer to the paths of $Y$ used in the calculation of $\tau$ as training paths and to those in the estimation of $E[X_\tau]$ and related quantities as testing paths.

In our numerical implementations we focus on two well-established methods for calculating the approximately optimal stopping time $\tau$,
the least-squares Monte Carlo algorithm of Tsitsiklis and Van Roy \cite{ TsVr} and the mesh method of Broadie and Glasserman \cite{BroGla04}. In the implementation of the Tsitsiklis-Van Roy method,
we use as basis functions at time $j$ the monomials up to second-order in the stocks $Y_{t_j}^d$ and the payoff from immediate exercise $X_j$. For the mesh method, we follow the implementation 
in \cite{belomestny2013pricing}, including the use of control variates, and omit the details here.

Before we come to the applications, let us briefly discuss these two methods: Arguably, the greatest strength of least-squares Monte Carlo methods is that 
even generic implementations, like the one above, often 
achieve approximations within a few percent of the true value at very low computational costs. 
However, for each fixed choice of basis functions, the methods'
bias can only be reduced down to some fixed level by increasing the number of training paths.
There is typically no practicable, generic method for controlling the precision without substantial tuning of the 
algorithm (the choice of basis functions) and/or massive increases in computational effort. 

We stick here to a fast, ``vanilla'' implementation of the algorithm which, in our view,
has all the advantages of least-squares Monte Carlo. In particular, this implementation is well-suited as an easy-to-evaluate quasi-control variate in Section \ref{secQVC}.\footnote{
For  reasons of numerical performance, we abstain from using European prices as basis functions since these are expensive to evaluate 
in high dimensions. For the same reason, we use the Tsitsiklis-Van Roy method instead of the slightly more popular Longstaff-Schwartz algorithm \cite{LS}. 
When working with two sets of paths, training and testing, the two methods typically lead to almost indistinguishable results, see \cite{B1}. }

In contrast, for the mesh method the bias can be controlled to arbitrary precision by increasing the number of training paths. No ``clever idea'', e.g. a choice of basis functions, is necessary.  
The price to pay for this considerable advantage is that evaluating the stopping times is highly expensive: When using the mesh method, the bulk of the computational effort is taken up by
calculating the realizations of $\tau$ along the simulated paths, since each path has to be compared with all training paths at each point in time.
In fact, the examples of Sections \ref{secQVC} and \ref{secML} can both be interpreted as attempts at increasing the mesh method's applicability by developing generic and
efficient control variates.

 \subsection{Assessing Parameter Uncertainty}\label{PU}
 In our first numerical example, there is a genuine interest in comparing two stopping times. We study the problem of estimating the sensitivity of stopping times estimated by the Tsitsiklis-Van Roy
 method to parameter misspecifications, namely, a misspecified volatility. Denote by $\tau^\sigma$ the stopping time calculated from training paths which have the correct volatility $\sigma$,
while $\tau^{\hat{\sigma}}$ is calculated from training paths with volatility $\hat{\sigma} \neq \sigma$. We wish to estimate the costs of exercising the option based on such a misspecified calculation, i.e.,
\[
\Delta({\hat{\sigma}})=E[X_{\tau^{{\sigma}}}-X_{\tau^{\hat{\sigma}}}]
\]
where the expectation is, of course, taken with respect to the correct model with volatility $\sigma$. 

Following an example in \cite{andersen2004primal, BroGla04} , we assume $10$ exercise dates $t_0,...,t_9=T$ which are equally distributed over the time horizon
$[0,T]$, and work with the following set of parameters: $d=2$, $T=3$, $r=0.05$, $\delta=0.1$, $\sigma=0.2$ $K=100$, $Y_0^d=90$. We use 100000 training paths of the underlying
Brownian motion for calculating the stopping
times and keep these fixed throughout.
Table \ref{tab:t33} reports estimates of the parameters $\rho_i$ and $v_i$  for different values of the misspecified volatility $\hat{\sigma}$.

\begin{table}
\renewcommand{\arraystretch}{1.4}
\centering
\begin{tabular}{|c|rrrr|}
\hline
 $\hat{\sigma} - \sigma $& 0.005& 0.01 & 0.015&  0.02 \\
 \hline
$\Delta(\hat{\sigma}) $ & 0.011  & 0.026&0.043 &0.066   \\
$E\left[ X_{\tau^{{\sigma}}} \right]$ & 8.042 & 8.042&8.042 & 8.042  \\
$E\left[ X_{\tau^{{\hat{\sigma}}}} \right]$ & 8.031 &8.016& 7.999 &7.976   \\
\hline
$P(\tau^\sigma \neq \tau^{\hat{\sigma}})$& 0.022& 0.043&0.062 &0.081   \\
$\rho_1$& 7.975 & 7.974& 7.972&7.972  \\
$\rho_2$& 0.053& 0.104&0.154 &0.199   \\
$v_1$&  0.008 &0.020& 0.037&0.061  \\
$v_2$& 4.023  &8.016& 12.053&16.066  \\
$R^*$& 271.8  & 176.0& 129.4 &103.3   \\
$\gamma^*$& 0.016 & 0.026& 0.037 & 0.047  \\
\hline
speed-up &  62.5& 38.5 & 27.0 &21.3   \\
\hline
\end{tabular}

\caption{Estimated simulation parameters for different values of $\hat{\sigma}$.
The speed-up $1/\gamma^*$  from using Nested CMC is given in the last row.}\label{tab:t33}
\end{table}

The first thing to observe from the table is that in all four cases Nested CMC leads to a substantial variance reduction, varying between a factor of about 60 and about 20, with the largest
gains if $\sigma$ and $\hat{\sigma}$ are most similar. The ratio between $v_1$ and $v_2$ is fairly constant and (much) smaller than the ratio between $\rho_2$ and $\rho_1$
which is thus decisive for the achieved variance reduction. We also report the probability that the two stopping times differ -- so that the subsimulations actually have to be carried out --
and find that it lies between 2\% and 8\%. These numbers are one key reason for the small values of $\rho_2$ and the high optimal numbers of subsamples (between 103 and 272). 

The units in which we report the $\rho_i$ are irrelevant (only the ratios matter) and chosen deliberately (but fixed across the table). More importantly,
it should be emphasized that the exact values of these numbers inevitably vary across different numerical implementations of the method.
In the above table, we tried to estimate these parameters as accurately as possible. Yet, as was shown in Section \ref{secCali}, this is not necessary -- rough estimates are sufficient  in practice.

\subsection{Improved Quasi-Control Variates}\label{secQVC}

In this section and the next, we turn to the more classical problem of calculating $E[X_{\tau^A}]$ for a given stopping time $\tau^A$. Introduce a second stopping time $\tau^B$
and write
\[
E[X_{\tau^A}]=E[X_{\tau^B}] + E[X_{\tau^A} - X_{\tau^B}].
\]
In a classical control variate approach, one would choose $\tau^B$ such that the first expected value on the right hand side can be calculated explicitly and would then estimate only
the second one by Monte Carlo. A popular example is the choice $\tau^B=J$ which corresponds to using European options as control variates.\footnote{There is a slight subtlety here, depending on whether
one uses $X_J$ or $E[X_J|\mathcal{F}_{\tau^A}]$ as a control variate, the European payoff or the European price. 
The second, superior choice can be understood as a conditional Monte Carlo estimator with $R=\infty$. This yields, however, typically a larger value of $v_1$ than the control variates we consider 
since $\tau^B=J$ is not necessarily a good approximation of $\tau^*$. Moreover, European prices are not always available in closed form.}
In a more abstract setting, Emsermann and Simon \cite{ES} pointed out that 
it can sometimes be beneficial to work with so-called quasi-control variates, i.e.,  a control variate whose expected value is not known explicitly.\footnote{We are unaware of previous applications of Quasi-Control Variates in Bermudan/American option pricing. As the present example shows, they are reasonably generic and
powerful variance reduction techniques for this important problem.} 
This is the case if $X_{\tau^B}$ is significantly 
cheaper to simulate than $X_{\tau^A}$. In that case,  the first summand can be  estimated with many (cheap) simulations of $X_{\tau^B}$. For the second summand, only a small number of (expensive) paths
may be necessary due to the variance reduction effect of $X_{\tau^B}$.
Finally, note that the second expectation is exactly the type of term which can be estimated efficiently by Nested CMC. Thus, we can hope
to significantly enhance quasi-control variates by our method. 

We retain the numerical example of the previous section, except that we increase the dimension to $d=3$. We choose $\tau^A$ as an approximate optimal stopping time calculated by the mesh method
with control variate from \cite{belomestny2013pricing}
with $2500$ training paths. As $\tau^B$ we choose a Tsitsiklis-Van Roy stopping time with $100,000$ training paths. In Table \ref{t:br} we state estimates of the expected values of $\mu^A=E[X_{\tau^A}]$, 
the variance $v^A=\textnormal{Var}(X_{\tau^A})$ and the cost $\rho^A$ for generating a sample of $X_{\tau^A}$, as well as the corresponding quantities for $\tau^B$. As expected, $v^A$ and 
$v^B$ are similar, but $\rho^A$ is by a factor \textit{three thousand} larger than $\rho^B$. Note also that $\mu^A$ is considerably larger than $\mu^B$. Since both estimates have a downward bias, this reflects
the greater accuracy of the mesh method. In particular, $\mu^A$ lies within the $95\%$ confidence interval $[11.265, 11.308]$ for $E[X_{\tau^*}]$ from \cite{andersen2004primal}.

Denote by $\rho(R)$ and $v(R)$ the computational costs and the variance per testing path when estimating $E[X_{\tau^A} - X_{\tau^B}]$ by Nested CMC with $R$ replications, i.e.,
\[
\rho(R)= \rho_1 +R\,\rho_2 \;\;\; \text{and}   \;\;\; v(R)= v_1+ \frac{v_2}{R}
\]
where the $\rho_i$ and $v_i$ are defined exactly as in Section \ref{secAlgo}. 
Let $N^B$ be the number of paths used to estimate $\mu^B$ and let $N$ be, as before, the number of
paths used in the estimation of $\mu^A-\mu^B$. For a given computational budget $C$ and fixed $R$, the optimal choice of $N^B$ and $N$ is given as the solution of
\[
\min_{N^B,\,N}\;\; \frac{v^B}{N^B}+ \frac{v(R)}{N}\;\;\;\text{s.t.}\;\; \rho^B\, N^B + \rho(R) N \leq C.
\]
By a calculation similar to those in \cite{ES} (or the proof of our Proposition \ref{propRstern}) it follows that the optimal ratio between $N$ and $N^B$ is given by
\begin{equation}\label{pathsQCV}
\frac{N^B}{N}=   \sqrt{\frac{v^B }{v(R) } \, \frac{\rho(R)}{\rho^B}}.
\end{equation}
regardless of the size of the computational budget.
Finally, observe that the optimal value $R^*$ of $R$ is the same as in Section \ref{secCali} regardless of how we 
allocate computational effort between the estimations of $\mu^B$ and $\mu^A-\mu^B$. From the values of the $v_i$ and $\rho_i$ we note that $R^*=97.037$ and the estimated gain $\gamma^*$ in the simulation
of $\mu^A-\mu^B$ is given by $\gamma^*=0.067$, corresponding to a speed-up of almost fifteen times in this part of the estimation.

\begin{table}
    \renewcommand{\arraystretch}{1.5}
\begin{center}
\begin{tabular}{|c|c|c|c|c|c|c|c|c|c|}
\hline
$\mu^A$ & $v^A$ & $\rho^A$ & $\mu^B$ & $v^B$ & $\rho^B$ & $v_1$ & $\rho_1$ & $v_2$ & $\rho_2$\\
\hline
 11.276&182&37.92&11.224&206&0.0124&0.044&
 36.23
 &19.536&
 1.728\\
\hline
\end{tabular}
\end{center}
\caption{Estimated simulation parameters.}\label{t:br}
\end{table}
Table \ref{t:alledrei} compares the performance of three Monte Carlo estimators for $\mu^A$ which have (approximately) the same computational costs and with parameters guided by the above considerations.
The first line gives the variance of a direct Monte Carlo
estimator of $\mu^A$ with $3,150$ sample paths. The second line shows the variance of a simple quasi-control variate estimator ($R=1$) with $N^B=536,178$ paths in the estimation of
$\mu^B$ and $N=2,989$ paths in the estimation of $\mu^A-\mu^B$. The third line shows the variance of a quasi-control variate estimator with $R=100$ replications in each of the 
$N=468$ paths in the estimation of $\mu^A-\mu^B$ and $N^B=1,784,813$ paths in the estimation of $\mu^B$. We thus see an improvement of more than a factor 100, which comes in equal parts from 
the quasi-control variate and from including nested simulations.

\renewcommand{\arraystretch}{1.5}
\begin{table}
   
\begin{center}
\begin{tabular}{|c|c|c|}
\hline
Method & Variance& Running time\\
\hline
\hline
Simple Monte Carlo & $60.4\times 10^{-3}$ & $125s$ \\
\hline
Quasi-Control Variate & $6.77\times 10^{-3}$  &$121s$\\   
\hline
Quasi-Control Variate with Nested CMC&  $0.514\times10^{-3}$ &$105s$\\ 
\hline
\end{tabular}
\end{center}
\caption{Comparison of the three methods with similar running times. This table reports averages over hundred runs of the simulation implemented in C++ on a standard system with a 2.6 GHz AMD processor.} \label{t:alledrei}
\end{table}
 
In the present example, the ratio between $v_1$ and $v_2$ is far more favorable than the ratio between $\rho_1$ and $\rho_2$, implying that the latter ratio governs the variance  
 reduction we achieve. This is due to the relatively high value of $\rho_1$ which arises since in about 60\% of cases it is the cheap Tsitsiklis-Van Roy stopping time which stops first.
 One can construct an even more efficient quasi-control variate by modifying the Tsitsiklis-Van Roy stopping time to be \textit{slightly} biased towards late stopping, thus increasing the
 variance $v_1$ but decreasing $\rho_2$. This can be achieved, e.g., by adding a small constant to the estimated continuation values.

 \subsection{An Improved Multilevel Algorithm}\label{secML}
 
Multilevel Monte Carlo methods, initiated by Heinrich \cite{heinrich2001multilevel} and Giles \cite{Gi}, can easily be understood as an extension of the quasi-control variate approach: 
Instead of a single quasi-control variate, there is a sequence of random variables, where each element in the sequence serves as a quasi-control variate to its successor. 
In a recent paper, Belomestny, Dickmann and Nagapetyan \cite{belomestny2013pricing} introduced and analyzed a multilevel method for our problem of numerically evaluating approximations of $E[X_{\tau^*}]$.
  
The algorithm of \cite{belomestny2013pricing} can be summarized as follows: Fix a number of levels $L$ and an increasing sequence $k_0, \ldots, k_L$ and consider the sequence of stopping times $(\tau(k_i))_i$ which are approximations of $\tau^*$ calculated
by the mesh method with $k_i$ training paths. Of these stopping times, $ \tau(k_L)$ is both, the most accurate and the most expensive. Thus, in order to estimate $E[X_{\tau(k_L)}]$ we write
\begin{equation}\label{ML}
E[X_{\tau(k_L)}] = E[X_{\tau(k_0)}] + \sum_{i=1}^L E[X_{\tau(k_{i})}-X_{\tau(k_{i-1})}]
\end{equation}
In the multilevel method of  \cite{belomestny2013pricing}, each of the summands on the right hand side is estimated independently by Monte Carlo with $N_i$ testing paths, where 
the sequence $N_0$, \ldots $N_L$ is \textit{decreasing}. The motivation for the algorithm is as follows:
$\tau(k_0)$ is cheap to evaluate and $E[X_{\tau(k_0)}]$ can thus be estimated with many testing paths. The summands in \eqref{ML} become more expensive as $i$ increases. Yet since their
contribution to the overall estimate is comparatively small, one can afford to estimate them with fewer sample paths.

For this method (under a suitable choice of parameters), it was shown in \cite{belomestny2013pricing} that  the overall computational effort (training and testing) required for a mean squared error of $\varepsilon$
behaves like $\varepsilon^{-2.5}$ as $\varepsilon$ gets small. This is a significant improvement over a simple Monte Carlo method with   $\varepsilon^{-3}$. In the following, we demonstrate
that applying Nested Conditional Monte Carlo in the estimation of $E[X_{\tau(k_{i})}-X_{\tau(k_{i-1})}]$ can lead to a substantial speed-up in the non-asymptotic regime.
For an entirely different application, the idea of using splitting techniques to speed-up multilevel Monte Carlo is mentioned as a possibility
already in \cite{Gi}. We are, however, unaware of later research which followed this suggestion. 

We retain the numerical example of the two previous sections but increase the dimension to $d=5$. We work with two levels, $L=2$, and $(k_0,k_1,k_2)=(100,1000,10000)$. The three stopping times
$\tau(k_i)$ are calculated by the mesh method with European control variate, specified exactly as in \cite{belomestny2013pricing}. Besides the number of training paths $k_i$, we also
increase the approximation quality of the numerical integration in the European control variate across levels, choosing precision parameters $(u_0,u_1,u_2)=(0.5,0.05,0.005)$, see \cite{belomestny2013pricing} for details. 
The training paths in the construction of $\tau(k_0)$ and $\tau(k_1)$ are subsets of the training paths for the true target stopping time $\tau(k_2)$, so that the additional effort 
from working with three instead of one stopping times is negligible at the training stage.
In light of \eqref{ML}, we can apply our Nested CMC twice, and obtain two sets of parameters $\rho_i$ and $v_i$ which are summarized in Table  \ref{tab:mlpar}.

\begin{table}
\renewcommand{\arraystretch}{1.4}
\centering
\begin{tabular}{|l|rr|}
\hline
 Level $i$ &1&2\\
\hline
$E[X_{\tau(k_{i})}-X_{\tau(k_{i-1})}]$ & 0.886& 0.026		 \\
$\rho_1$& 9.8& 111.3 \\
$\rho_2$ &1.4 & 1.8  \\
$v_1$ &  2.292& 0.037 \\
$v_2$& 55.429&14.485  \\
$\gamma^*$ &0.28& 0.031\\
$R^*$&   13.018 & 158.707\\
\hline
\end{tabular}
\caption{
Simulation parameters at the two levels.
 }\label{tab:mlpar}
\end{table}

We thus see, that Nested CMC leads to drastic speed-up of about a factor 32 at the high-precision level $i=2$, and to a still decent one of about 3.5 at the intermediate level $i=1$.\footnote{This effect would become more pronounced 
if we included further levels of higher precision: These have an even greater speed-up factor from using Nested CMC.}
At the ``base level'', i.e., the calculation of  $E[X_{\tau(k_0)}]=15.698$ we have a variance of $\textnormal{Var}(X_{\tau(k_0)})= 251.3$ and a cost per sample which we normalize to 1.  
Following \eqref{ML},  the expected value we are calculating is thus 
\[
E[X_{\tau(k_L)}] = 15.698+0.886+0.026 = 16.610
\]
which is well within the confidence interval $[16.60, 16.66]$ for $E[X_{\tau^*}]$ from  \cite{andersen2004primal} for this example. In determining the optimal number of testing paths for 
estimating each summand in $\eqref{ML}$, we use the following generalization of \eqref{pathsQCV} found, e.g. in \cite{Gi,belomestny2013pricing}: For a fixed computational budget,
the \textit{squared} number of paths $N_i^2$ in the estimation of each summand should be proportional to the variance divided by the costs per sample for this summand.

\begin{table}
\renewcommand{\arraystretch}{1.4}
\centering
\begin{tabular}{|c|ccc|c|}
\hline
Method  &  & & & Variance\\
\hline
      & &$N$  & & \\
Simple Monte Carlo & &1790  & &0.13  \\  
\hline
 & $N_0$& $N_1$ & $N_2$ &  \\
 Multilevel & 38760 & 5550& 880 & 0.033\\
 Multilevel with Nested CMC & 86780 & 2650 & 100& 0.0067 \\
\hline
\end{tabular}
\caption{Overall expected variances of the three methods with identical expected computational costs. 
}\label{tab:mlvar}
\end{table}
Table \ref{tab:mlvar} compares Multilevel Monte Carlo with and without nesting for a fixed expected computational budget of 200000 time units.
As suggested by Table \ref{tab:mlpar}, we use $R=13$ and $R=159$ replications
in the Nested CMC algorithms at the two levels.

We also present results for a simple Monte Carlo estimator of the same expected value, $E[X_{\tau(k_2)}]$, under the same budget. Simple Monte Carlo has a cost per sample of 112.9 and $\textnormal{Var}(X_{\tau(k_2)})=234.1$. 
There is a variance reduction by a factor 19.6 between simple Monte Carlo and Multilevel Monte Carlo with Nested CMC, the larger part of which (a factor 5) comes from 
incorporating the nested simulations. 

\section{Conclusion} \label{secCon}

In this paper, we have introduced Nested Conditional Monte Carlo, a simple Monte Carlo method for estimating the difference between two stopped versions of the same stochastic process. 
The algorithm is easy to calibrate by estimating two variance and two running-time parameters. Moreover, rough parameter estimates provably suffice for a near-optimal performance of the method.
We demonstrated that besides its direct applications our method can be used as a generic tool for enhancing variance reduction methods for stopping problems.
In fact, we can see little reason why one should implement the Quasi-Control Variate or Multilevel Monte Carlo methods of Section \ref{secApp} without including subsimulations.
The resulting variance reduction methods are very efficient and, unlike classical control variates, do not require that anything can be computed explicitly.

In terms of applications, we have focused on Bermudan option pricing but many other fields of application are conceivable. As examples, consider credit risk modelling, where events of default and distress 
are often modelled by stopping times, or pricing heuristics in revenue management as in \cite{feng1995optimal} where the near-optimal timing of sales and promotions is studied.

In our algorithm, we use the same number of replications on each path. Yet, as demonstrated in a different type of application in \cite{broadie2011efficient}, 
numerical efficiency can be improved by allocating more replications to ``critical'' trajectories. An extension of our method which achieves this -- while retaining unbiasedness --
splits the trajectories at \textit{every} time point between $\tau^\wedge$ and $\tau^{\vee}$. In this way, trajectories with a large value of $\tau^\vee-\tau^\wedge$
are automatically investigated more intensively. Alternatively -- if one is not concerned about a small bias -- combinations of our method with Importance Sampling might be fruitful.
We leave these and further extensions and applications of our method to future research.

 \appendix
\section{Proofs} 
 
 \begin{proof}[Proof of Proposition \ref{propUB}]
 To see the unbiasedness, note that
 \begin{eqnarray*}
 E[\Delta^{(N,R)}]&=&\frac{1}{N} \sum_{i=1}^N \frac{1}{R} \sum_{r=1}^R E[S^{(i)} (X^{(i,r)}_{\tau^{\vee,(i,r)}} - X^{(i)}_{\tau^{\wedge,(i)}})]\\
 &=& \frac{1}{N} \sum_{i=1}^N \frac{1}{R} \sum_{r=1}^R E[X_{\tau^A}-X_{\tau^B}]=\Delta
 \end{eqnarray*}
  where the second equality simply used that the term inside the expectation is an independent copy of $X_{\tau^A}-X_{\tau^B}$. For the variance, note first that the outer sum over $i$
  is a sum of independent, identically distributed random variables and thus
 \begin{eqnarray*}
 \textnormal{Var}(\Delta^{(N,R)}) =\frac{1}{N} \textnormal{Var}\left( \frac{1}{R} \sum_{r=1}^R  S^{(1)} (X^{(1,r)}_{\tau^{\vee,(1,r)}} - X^{(1)}_{\tau^{\wedge,(1)}}) \right).
  \end{eqnarray*}
  Applying the conditional variance decomposition formula yields  \[\textnormal{Var}(\Delta^{(N,R)})=\frac{v_1}{N} + \frac{v_2}{R N}\] with 
  \[
  v_1= \textnormal{Var}\left(E\left[\left. \frac{1}{R} \sum_{r=1}^R  S^{(1)} (X^{(1,r)}_{\tau^{\vee,(1,r)}} - X^{(1)}_{\tau^{\wedge,(1)}})\right| \mathcal{F}^{{\tau^{\wedge},(1)}}   \right] \right)
  \]
  and
  \[
  v_2= R\cdot E\left[ \textnormal{Var}\left(\left. \frac{1}{R} \sum_{r=1}^R  S^{(1)} (X^{(1,r)}_{\tau^{\vee,(1,r)}} - X^{(1)}_{\tau^{\wedge,(1)}})\right| \mathcal{F}^{{\tau^{\wedge},(1)}}  \right)\right]. 
  \]
  and it remains to see that these values of $v_1$ and $v_2$ coincide with those in the proposition. 
  Note that the summands are independent and identically distributed conditionally on  $\mathcal{F}^{{\tau^{\wedge,(1)}}}$. For $v_1$ this implies that 
  $$E\left[\left.  S^{(1)} (X^{(1,r)}_{\tau^{\vee,(1,r)}} - X^{(1)}_{\tau^{\wedge,(1)}})\right| \mathcal{F}^{{\tau^{\wedge,(1)}}}  \right]$$ does not depend on $r$ and thus
   \[
  v_1= \textnormal{Var}\left(E\left[\left.    S^{(1)} (X^{(1,1)}_{\tau^{\vee,(1,1)}} - X^{(1)}_{\tau^{\wedge,(1)}})\right| \mathcal{F}^{{\tau^{\wedge},(1)}}  \right] \right).
  \]
 A similar argument for $v_2$ now yields
  \[
  v_2= E\left[ \textnormal{Var}\left(\left.  S^{(1)} (X^{(1,1)}_{\tau^{\vee,(1,1)}} - X^{(1)}_{\tau^{\wedge,(1)}})\right| \mathcal{F}^{{\tau^{\wedge},(1)}}  \right)\right].
  \]
  Noting that $S^{(1)} (X^{(1,1)}_{\tau^{\vee,(1,1)}} - X^{(1)}_{\tau^{\wedge,(1)}})$ and $\mathcal{F}^{{\tau^{\wedge,(1)}}}$ are copies of $X_{\tau^A}-X_{\tau^B}$ and $\mathcal{F}_{\tau^\wedge}$
  allows to conclude the proof. Finally, let us emphasize that the above argument does take into account the fact that on \textit{some} trajectories -- those trajectories where $\tau^A$ and $\tau^B$ coincide --
  $X_{\tau^A}-X_{\tau^B}$ is $\mathcal{F}_{\tau^\wedge}$ measurable.
  \end{proof}

  \begin{proof}[Proof of Proposition \ref{propRstern}]
  Since the objective function decreases in both $R$ and $N$, it is clear that the budget constraint holds with equality at the optimum, $c(N,R)=C$. Solving the constraint for $N$ and substituting
the result  into the objective yields
\[
\min_{R }  \frac{1}{C} \left(  v_1 (\rho_1 +\rho_2 R )+ v_2 \frac{\rho_1 +\rho_2 R }{R}\right)\;\;\; \textnormal{s.t.}\;\;\; R \geq 1.
\]
Using that the minimization problem is invariant to monotone transformations, we can write this as
\[
\min_{R } \frac{v_1 \rho_2}{v_2 \rho_1} R +   \frac{1}{R} \;\;\; \textnormal{s.t.}\;\;\; R \geq 1.
\]
Clearly, the solution to this convex minimization problem is $R^*=\max(1, R')$ where $R'$ is the solution of the associated unconstrained minimization problem which is given by
$
R' = \sqrt{\frac{v_2 \rho_1}{v_1 \rho_2} } .
$
\end{proof}

 \begin{proof}[Proof of Proposition \ref{propgain}]
The formula for $\gamma^*$ follows with a few algebraic manipulations after substituting $R^* = \sqrt{\frac{v_2 \rho_1}{v_1 \rho_2} }$ into $V$. We turn to the lower bound. 
By symmetry, it suffices to prove that for all positive real numbers $a$ and $b$ with $ab\leq 1$
\[
\frac{(a+b)^2}{(1+a^2)(1+b^2)} \geq \frac{a^2}{1+a^2}.
\] 
To see this, note that we can bound the numerator as follows
\[
(a+b)^2 \geq a^2+ ab  \geq a^2+  a^2b^2 =a^2(1+b^2). 
\] 
For the upper bound it suffices to observe that
\begin{eqnarray*}
 \frac{(a+b)^2}{(1+a^2)(1+b^2)} &\leq& \frac{2 a^2}{(1+a^2)(1+b^2)}+\frac{2 b^2}{(1+a^2)(1+b^2)}\\
 &\leq& \frac{2 a^2}{1+a^2}+\frac{2 b^2}{1+b^2}\\
 &\leq& 4 \max\left( \frac{a^2}{1+a^2}, \frac{b^2}{1+b^2}\right).
\end{eqnarray*}
where we used in the first step that $(a+b)^2 \leq 2(a^2+b^2)$.
\end{proof}

 \begin{proof}[Proof of Proposition \ref{propAlpha}]
  Observe first that we can write 
  \[
  V(R)=\rho_1 v_1+\rho_2 v_2+\rho_2 v_1 \left(R+\frac{{R^*}^2}{R} \right)
  \]
  and
  \[
  V(\alpha R^*)=\rho_1 v_1+\rho_2 v_2+\rho_2 v_1 (\alpha + \alpha^{-1}) R^*.
  \]
 Therefore we have $V(\alpha R^*)=V(\alpha^{-1}R^*)$ and by convexity $V(R)\leq V(\alpha R^*)$ for all $R$ in the interval. It thus suffices to prove the upper bound for 
 \[
 \frac{V(\alpha R^*)}{V(R^*)}= \frac{\Gamma + \alpha + \alpha^{-1}}{\Gamma + 2}\;\;\;\textnormal{ where }\;\;\; \Gamma =\frac{\rho_1 v_1+\rho_2 v_2}{\rho_2 v_1 R^*}.
 \]
 Since $\alpha +\alpha^{-1} \geq 2$, we can bound this expression from above by replacing $\Gamma$ with a smaller number. In particular, $\Gamma \geq 2$ yields the desired inequality
 \[
 \frac{V(\alpha R^*)}{V(R^*)}\leq  \frac{2 + \alpha + \alpha^{-1}}{4}.
 \]
 To see that we indeed have $\Gamma \geq 2$,
 note that by inserting the expression for $R^*$ we can write
 \[
 \Gamma = \sqrt{\frac{\rho_2 v_2}{\rho_1 v_1}}+\sqrt{\frac{\rho_1 v_1}{\rho_2 v_2}}.
 \]
 $\Gamma \geq 2$ now follows from the fact that $x+x^{-1} \geq 2$ for all $x \geq 0.$
  \end{proof}

 \begin{proof}[Proof of Corollary \ref{propR}]
  In the proof of Proposition \ref{propAlpha} we saw that $V(\alpha R^*)=V(\alpha^{-1} R^*)$. For $\alpha=R^*$ this gives $V({R^*}^2)=V(1)$. Thus, (i) follows from the convexity of $V$.
  The argument for (ii) is similar.
  For (iii) note first that if $R^*<1.5$ we have $R^\#=1$ and nothing is to prove. 
  By (i) it thus suffices to show $R^\# < {R^*}^2$ for $R^* \geq 1.5$. To see this, note that
  $
  R^\# \leq R^*+\frac{1}{2} < {R^*}^2
  $
  where the last inequality holds for all $R^* > \frac{1+\sqrt{3}}{2} \approx 1.37$.
 \end{proof}

\bibliographystyle{plain}
\bibliography{branching2.bib}

\end{document}